\documentclass{article}

\usepackage{amsmath,amssymb,amsthm,commath,bbm,dsfont}
\usepackage{xifthen}
\usepackage[dvipsnames]{xcolor}
\usepackage[utf8]{inputenc} 
\usepackage[T1]{fontenc}    
\usepackage{hyperref}       
\usepackage{url}            
\usepackage{booktabs}       
\usepackage{amsfonts}       
\usepackage{nicefrac}       
\usepackage{microtype}      
\usepackage{xcolor}         
\usepackage{todonotes}

\usepackage{verbatim}
\usepackage{graphicx}
\usepackage{subcaption}
\graphicspath{ {./images/} }
\usepackage{aux/macros}

\usepackage[accepted]{aux/icml2022}

\begin{document}

\twocolumn[
\icmltitle{Learning to Hash Robustly, Guaranteed}

\begin{icmlauthorlist}
\icmlauthor{Alexandr Andoni}{yyy}
\icmlauthor{Daniel Beaglehole}{xxx,zzz}
\end{icmlauthorlist}

\icmlaffiliation{yyy}{Department of Computer Science, Columbia University, New York, New York}
\icmlaffiliation{xxx}{Computer Science and Engineering, UCSD, La Jolla, California}
\icmlaffiliation{zzz}{Work completed while at Columbia University}

\icmlcorrespondingauthor{Daniel Beaglehole}{dbeaglehole@ucsd.edu}

\icmlkeywords{Theoretical Computer Science, Theory of Machine Learning, Algorithms, Nearest Neighbors Search}
\vskip 0.3in
]

\printAffiliationsAndNotice{}

\begin{abstract}
    The indexing algorithms for the high-dimensional nearest neighbor search (NNS) with the best worst-case guarantees are based on the randomized Locality Sensitive Hashing (LSH), and its derivatives. In practice, many heuristic approaches exist to "learn" the best indexing method in order to speed-up NNS, crucially adapting to the structure of the given dataset.
    Oftentimes, these heuristics outperform the LSH-based algorithms on real datasets, but, almost always, come at the cost of losing the guarantees of either correctness or robust performance on adversarial queries, or apply to datasets with an assumed extra structure/model. In this paper, we design an NNS algorithm for the Hamming space that has worst-case guarantees essentially matching that of theoretical algorithms, while optimizing the hashing to the structure of the dataset (think instance-optimal algorithms) for performance on the minimum-performing query. We evaluate the algorithm's ability to optimize for a given dataset both theoretically and practically. On the theoretical side, we exhibit a natural setting (dataset model) where our algorithm is much better than the standard theoretical one. On the practical side, we run experiments that show that our algorithm has a 1.8x and 2.1x better recall on the worst-performing queries to the MNIST and ImageNet datasets. 
\end{abstract}

\section{Introduction}

In the nearest neighbor search (NNS) problem, we are to preprocess a dataset of points $P$ so that later, given a new query point $q$, we can efficiently report the closest point $p^*\in P$ to $q$. The problem is fundamental to many high-dimensional geometric tasks, and consequently to modern data analysis, with applications from computer vision to information retrieval and others \citep{NNNIPS}. See surveys \citep{wang2015learning,andoni18-icm}.

Depending on whether the algorithm has worst-case theoretical guarantees, the indexing solutions for the NNS problem are essentially split into two categories. The first category of algorithms, with theoretical guarantees, are usually based on randomized space partitions, namely Locality-Sensitive Hashing (LSH), and its derivatives---conceptually similar to the random dimension reduction \citep{JL}. In order to provide a worst-case guarantee, one focuses on the $c$-approximate version, for some approximation $c>1$, where one has to report a point $p\in P$ at distance at most $cr$ as long as $\|q-p^*\| \le r$. For example, in the case of the $d$-dimensional Hamming space $\mathbb{H}^d =\{0,1\}^d$, the original LSH paper \citep{HIM12} gives an algorithm with $O(n^\rho d)$ query time and $O(n^{1+\rho}+nd)$ space where $\rho=1/c$, which is optimal for LSH algorithms \citep{OWZ11}. Crucially, the algorithm guarantees that, if there exists a point $p^*$ at distance at most $r$, then the data structure returns a point at distance at most $cr$ with probability at least, say, $90\%$ {\em over the randomness of the algorithm} (termed success probability).

Algorithms from the second category are based on the idea of finding (learning) {\em the best possible} space partition (hashing) for the given dataset, which, in practice, is usually "nicer" than a worst-case one. For example, PCA trees use partitions based on the Principal Component Analysis of the dataset \citep{Sproull-pcaTree, mcnames2001fast,verma2009spatial,aakk-spectral-14,keivani2018improved}, although many more methods exist; see survey \citep{wang2015learning} for some of them as well as more recent \citep{Dong2020Learning}. While usually more efficient in practice, such algorithms come at the cost of losing the worst-case guarantees. Most often, the correctness is not guaranteed per query: there are (adversarial) queries on which the data structure will fail. Alternatively, the runtime may devolve into a (na\"ive) linear scan. To address such issues, one approach has been to prove guarantees assuming the dataset has extra structural properties: e.g., that it has low doubling dimension, or that it is generated according to a random model.

Bridging the gap between these two categories of algorithms has been recognized as a big
open question in Massive Data Analysis, see e.g. the National Research Council report [Section 5] \cite{NRC-report} in the closely-related setting of random dimension reduction. We summarize this challenge as the following "instance optimality" question:

\begin{challenge}
\label{cha:main}
    Develop NNS algorithms that adapt optimally to the input dataset, while retaining provable guarantees for all, including adversarial, queries.
\end{challenge}

We address the above challenge in this paper. Before delving into our specific results, we comment on two non-answers. First, a recent line of research led to {\em data-dependent hashing} algorithms that similarly have worst-case guarantees \citep{AINR-subLSH,AR-optimal,AILSR15}, improving, for example, the original exponent $\rho$ of \citep{HIM12} to $\rho=\tfrac{1}{2c-1}+o(1)$. While this line of work shows that adapting to the dataset can improve the  performance for a worst-case dataset, it does not seek to improve the performance further if the dataset is "nice". Second, a straight-forward solution to the challenge could be to run both a practical heuristic and a theoretically-guaranteed algorithm (timing out the latter one if needed). Such a solution however still does not seek to improve the performance for all, especially adversarial, queries.\footnote{In particular, that would merely split the queries into two classes: those on which the heuristic is successful with improved performance, and those on which it is not and hence the performance is that of a worst-case theoretical algorithm.}

We also note that it generally seems hard to adapt the heuristic algorithms to have theoretical guarantees for all queries. Most such algorithms learn the best partition, yielding a deterministic index\footnote{While some use randomization, it is usually used to find the optimal partition (e.g., via SGD), but not to randomize the partition itself.}---i.e., building a few indexes does not help failed queries (in contrast to the LSH-based randomized indexes). At the same time, it is known that the deterministic algorithms are unlikely to yield worst-case guarantees \citep{PTW10}. In particular, it is usually possible (and easy) to construct an adversarial query, by planting it "on the other side" of the part containing its near neighbor - guaranteeing failure. Hence, a solution for the above challenge should involve randomized partitions (as LSH does).

\subsection{Our Results}
\label{sec:ourResults}

We address Challenge~\ref{cha:main} in the case of (approximate) NNS problem under the Hamming space\footnote{See discussion of the Euclidean space in Appendix~\ref{apx:discussion}.}, for which we design an algorithm that adapts to the dataset's potential structure, while maintaining the performance guarantees for all queries. Our algorithm
should be seen from the perspective of {\em instance optimal} algorithms: an algorithm that is the best possible, within a class of algorithms, for the given dataset. 

Our algorithm directly optimizes the performance for all possible queries, for the given fixed dataset. We obtain the following properties (see Theorem~\ref{thm:CorrectnessRuntime} in Section~\ref{sec:main}):
\begin{enumerate}
    \item Correctness: For any query $q$, the algorithm is guaranteed to return the $c$-approximate near neighbor with success probability at least $\Omega(n^{-\rho})$ for some $\rho\le 1/c$, the exponent obtained by the optimal LSH \citep{HIM12,OWZ11}. (Probability is over the randomness of the algorithm only.)
    
    \item Performance: The query time is $O(d^2)$ and the space is $O(n)$, and the preprocessing time is $O(n\cdot \poly(d))$. Note that, as is standard for LSH algorithms, we can boost the success probability to, say, 90\% by repeating the algorithm for $O(n^\rho)$ times, obtaining the usual tradeoff of $O(n^\rho\cdot \poly(d))$ query time and $O(n^{1+\rho}+nd)$ space overall (but for smaller $\rho$).
    \item Data-adaptive:
    The algorithm adapts to the input dataset, and can obtain better success probability for "nicer" datasets. In fact, under certain conditions, the algorithm is "instance optimally" adaptive to the dataset.
\end{enumerate}

We now discuss the last claim of data-adaptivity. The ideal goal would be to obtain an instance-optimal algorithm. Our algorithm becomes instance-optimal (in a precise sense described in the next section), if we are given optimal values for certain parameters $\rho$ during the construction.
Alas, we do not know how to compute these parameters efficiently (and thus do not achieve instance optimality).

Instead, we evaluate the last claim by showing that our algorithm achieves theoretical and practical improvements over the only other NNS algorithms with similarly strong guarantees for Hamming space (standard LSH indexes) for a range of parameters. 
On the theoretical side, we formulate a concrete model for the dataset, for which our algorithm improves on the success probability for all queries. We specifically consider the case where the dataset is a mixture model: it is composed of several clusters, where each point is generated iid. We note that our algorithm is not designed specifically for this model; instead it is a natural theoretical model for "nicer" datasets to evaluate improvement of an algorithm. See Section \ref{sec:mixture_model} and Section ~\ref{apx:improveUniProof} in supplemental material.

On the practical side, we run experiments that show that our algorithm has a 1.8x better recall on the worst-performing queries to the MNIST dataset, and a 2.1x better recall on the bottom tenth of queries to the ImageNet dataset. See Section~\ref{sec:experiments}.

\subsection{Technical Description of our Algorithm}
\label{sec:introTechnical}

We now give an overview of our main algorithm and the tools involved. Our algorithm is based on the LSH Forest method \citep{bawa2005lsh} for Hamming space, in which the dataset is iteratively partitioned according to the value in a coordinate, thereby progressing down the constructed tree. In particular, beginning with the entire dataset in the root of an LSH tree, in each node, we pick a random hash function and use it to partition the dataset. The partitioning stops once the dataset becomes of size $\le C$ for some constant $C$, termed stopping condition. Otherwise, we recurse on each new part (child of the current node in the tree).

The key new component of our algorithm is that, in each node, we {\em optimize for the best possible distribution} over hash functions, for the given dataset. In particular, in each node, we solve an optimization to produce a distribution $\pi$, over coordinates $[d]$, that maximizes the probability of success over all (worst) possible queries. Following this optimization, we draw a coordinate from the optimized distribution $\pi$, hash the dataset on the resulting coordinate (to produce two children corresponding to bits 0 and 1 at that coordinate). We then recurse on each of the hashed datasets (children) until the current dataset is less than a fixed constant (stopping condition). The formal algorithm, {\sc BuildTree}, is described in Alg.~\ref{alg:build_tree}. 

The main technical challenge is to compute the optimal distribution $\pi$, for which we use a two-player game to solve a min-max problem. Note that this is a question of efficiency---it is easy to compute the optimal $\pi$ in exponential time (and an instance optimal algorithm in general)---and hence our goal is to do so in polynomial time.
Specifically, our method directly optimizes for robustness by computing solutions to a min-max optimization. In particular, we seek the distribution over hashes that maximizes the recall on the minimum performing queries. What are the exact quantities we want to optimize at a given node? Consider a distribution $\pi \in \Delta[d]$ over hash functions, and a query/near neighbor pair $p,q \in \mathbb{H}^d$. The true instance optimal, min-max, objective function at each node is the following:
\begin{align*}
    &\Pr_{Alg}[success] = \\
    &\sum_{i\in[d]} \pi_i \Pr_{Alg}[success \mid \text{ hash coordinate $i$, bit $q_i$}] \cdot \mathbbm{1}\{p_i = q_i \}
\end{align*}
which is a function of the success probability on the remainder of the dataset, ${\Pr[success \mid \text{hash coord. $i$, bit $q_i$}]}$. However, exactly computing the probability of success on the remainder of the tree appears computationally intractable, as one would need to have considered all possible subsets of hashes (exponential in dimension). Instead, we approximate the recursive probabilities by using a {\em lower bound} on the probability of success in the remainder of the tree. In fact, there's already a natural candidate for such a lower bound: the success probability of the standard LSH, which hashes on uniformly-random coordinate(s). Hence, our optimization becomes as follows, where the maximum is over distributions $\pi \in \Delta[d]$, $n_{i,p_i}$ is the size of the dataset after hashing on coordinate $i$, bit $p_i$, and for a chosen parameter $\rho \in (0,1)$, the optimization program we solve is:
\begin{align}
	\max_{\pi\in\Delta[d]} \min_{\substack{p\in P\\q:\|p-q\|\le r}} \sum_{i=1}^d \pi_i n_{i,p_i}^{-\rho}  \mathbbm{1} \{p_i = q_i \} \label{eqn:minmaxopt}
\end{align}

Notably, if we set ${\rho = \rho_i(q)}$ for each ${i \in [d]}$, where ${n_{i,p_i}^{-\rho_i(q)}= \Pr[success \mid \text{hash coordinate $i$, bit $p_i$}]}$, we obtain exactly the instance optimal objective. We don't know the values of $\rho_i$'s but can use an upper bound instead: $\rho_i\le 1/c$. (In fact, one can compute directly an upper bound using any data-independent distribution $\pi$---e.g., even uniform distribution $\pi$ sometimes yields better estimates than $1/c$.)

We solve the min-max program from Eqn.~\eqref{eqn:minmaxopt} by finding a Nash equilibrium in an equivalent two-player zero sum game, in which the worst-performing queries are iteratively presented to a "player" who learns hash functions to maximize the success probability on those queries. The main question is under what circumstances can we find such a Nash equilibrium {\em efficiently}? In the case of our hash/query game, although there are exactly $d$ hash functions available to hash player, there are $n {d \choose r}$---potentially exponential in $d$---many query/NN pairs available to the query player. 

Nonetheless, it turns out we can approximately solve this game efficiently, in $n\cdot\poly(d)$ time! We use a repeated-play dynamic from \citep{FreundSchapire} in which the hash player performs the multiplicative weights update and the query player chooses the query that minimizes their loss on the hash distribution most recently played by the hash player. Indeed, while the complexity of the game is polynomial in the number of hash player strategies, it is essentially independent of the number of possible queries, as we have reduced the query player's complexity contribution to that of a single minimization (see details in Sec~\ref{section:minmax} and Sec.~\ref{section:qpm}, Supplemental Material).

\subsection{Other Related Work}

This paper focuses on {\em indexing} NNS algorithms, which can be contrasted to the {\em sketching} algorithms; see \citep{wang2015learning}. In the latter, the goal is to produce the smallest possible sketch for each point in order to speed-up a linear scan over the dataset (of sketches). Such solutions have a query time (at least) linear in $n$, in contrast to the indexing algorithms, which are {\em sublinear}, typically $n^\rho$ for $\rho<1$. Furthermore, one can often combine the two: use the indexing NNS algorithm to filter out all but a smaller set of candidate points and then use (preprocessed) sketches for faster distance evaluations on them \citep{wu2017multiscale,johnson2019billion}.  

We also note that there exist other practical NNS algorithm, which do not directly fit into the "learning to hash" paradigm alluded to before. For example, the algorithm from \citep{malkov2018efficient} builds a graph on the dataset, such that a future query will perform a graph exploration to reach the nearest neighbor. While very competitive in practice, it again provides no guarantees. It remains a formidable challenge to derive theoretical guarantees for such algorithms.

\section{Preliminaries}
\label{sec:prelims}

We label the dataset as $P \subset \{0,1\}^d = \mathbb{H}^{d}$ where $|P|=n$. Formally, we solve the $c$-approximate near neighbor problem, where, given a threshold $r>0$, and approximation $c>1$, we need to build a data structure on $P$ so that, given a query $q$, we return a point $p\in P$ with $\|q - p\|_1 \leq cr$, as long as there exists a point $p^*$ with $\|q - p^*\|_1 \leq r$. In that case, for the given $q$, we call such a point $p^*$ a near neighbor of $q$, and $p$ an approximate near neighbor.

\begin{definition}
\label{def:success}
For a given query $q$ and a near neighbor $p^*$, we consider an LSH tree to be successful on that pair if when the query algorithm halts on a node $v$, $q$ and $p^*$ are both in the bucket at node $v$. The probability with which it happens (over the randomness of the algorithm) is referred to as {\bf success probability}, denoted $\Pr[success]$.
\end{definition}

Our algorithm builds a tree top-down, from a node to its children partitioning the dataset according to the chosen hash function. For a node $v$, we use $P^v\subset P$ to denote the set of dataset points that reached the node $v$ (have been hashed to $v$ according to the hash function of the ancestors of $v$). We also call $P^v$ as the "bucket" at $v$, and let $n_v=|P^v|$.
Each (internal) node $v$ has an associated hash function used to partition $P^v$, which is described by the coordinate $i\in[d]$ by which we partition $P^v$. In particular, $P^v_{i,b}$ indicates the subset of datapoints in $P^v$ that have bit $b$ at coordinate $i$. The node $v$ splits $P^v$ into $P^v_{i,0}$ and $P^v_{i,1}$. We let $n_{i,v} := |P^v_{i,q_i}|$.

\begin{definition}
\label{def:balance}
A coordinate $i \in [d]$ is called \textbf{$\epsilon$-balanced} for the dataset $P^v$ and $0 \leq \epsilon \leq 0.5$ if:
\begin{align*}
\max(|P^v_{i,0}|, |P^v_{i,1}|) = (1-\epsilon) |P^v|.
\end{align*}
\end{definition}

For the analysis that follows, we make the trivial assumption that hashing is done without replacement (i.e. once a coordinate $i$ is used to hash, it is never used again in a tree descendant).

\paragraph{Notation.} For two vectors $x,y \in \mathbb{R}^d$, we denote their element-wise product by $x \odot y \in \mathbb{R}^d$. We denote the transpose of a vector $x$ by $x^\prime$. For a vector $x$, we denote its $i$-th coordinate by $x(i)$ or $x_i$. Let $e_i$ be the $i$-th standard basis vector.


\section{Main Algorithm} 
\label{sec:main}

We now present and analyze our LSH forest algorithm with hash functions adapting to the given dataset. We then show that our algorithm (1) is correct, and (2) has worst-case performance guarantees. We show our algorithm has improved performance in experiments in Section~\ref{sec:experiments} and Section~\ref{sec:AdditionalExperiments}, and on "nice" datasets in Section~\ref{sec:mixture_model} and Section~\ref{apx:improveUniProof} (Supplemental Material).

We present our pre-processing algorithm in Alg.~\ref{alg:build_tree}. The algorithm is an LSH forest algorithm, where, beginning with the entire dataset at the root, we construct the tree by performing a min-max optimization at the current node to compute the best distribution over hashes, picking a random hash function from this optimized distribution, and recursing on the hashed datasets until the datasets are of constant size. 

The main component of the algorithm is to compute the optimal distribution for the given node, described in Alg.~\ref{alg:MinMaxOpt}. Specifically, for this goal, we setup a min-max optimization, Eqn.~\eqref{eqn:minmaxopt}, which we solve efficiently by iterating a two-player zero sum game (see Section~\ref{section:minmax}). 

Our main correctness and worst-case performance guarantee is in the following theorem. We remark that the main algorithm requires an input parameter $\rho$, which we discuss, along with an interpretation of the probability guarantee, in  Section~\ref{section:succprobdiscussion}.

\begin{theorem} [Correctness and Runtime]
\label{thm:CorrectnessRuntime}
Fix stopping condition $C \ge 1$ to be a constant, and query algorithm parameter $m > 1$. Suppose there exists a $\rho\in(0,1)$, such that for any node $v$, there is a distribution $\pi_v$ over hash functions such that for any query/near neighbor pair $q, p^*\in\mathbb{H}^d$, both hashing into node $v$, such that fewer than $\frac{1}{m}$ fraction of the bucket $P^v$ are approximate near neighbors of $q$, $\mathbb{E}_{i \sim \pi_v}[\mathbbm{1}\{p^*_i = q_i \} \cdot f_{i,p^*_i,v}^{-\rho}] \geq 1$ where ${f_{i,p^*_i,v} = |P^v_{i,p^*_i}| / |P^v|}$. Then, using $\rho$ as the exponent parameter, Algorithm \ref{alg:build_tree} constructs a tree that satisfies:
\begin{align*}
	\Pr_{Alg}[\text{success on $q,p^*$}] \geq n^{-\rho}- 2\eps d,
\end{align*}
where $\epsilon > 0$ is the input parameter. Furthermore, $\rho\le \gamma/c$ for $\gamma = \frac{1}{1-1/m}$. 

The pre-processing time to construct a single tree as in Algorithm \ref{alg:build_tree} is $O\left(\frac{1}{\epsilon^2} n d^2 \ln^2 d \right)$, and the resulting query time by Algorithm \ref{alg:query_tree} is $O\left(m d^2 \right)$.
\end{theorem}

\begin{algorithm}[h!]
    \caption{Build Tree}
    \label{alg:build_tree}
\begin{algorithmic}[1]
    \STATE {\bfseries Input}: dataset $P^v$, exponent parameter $\rho$, stopping condition $C$, approximation $\eps$
    \STATE create an empty node $v$
    \STATE set $v$.dataset =  $P^v$
    \IF{$|P^v| > C$}
    \STATE set $\pi_v =$ MinMaxOpt($P^v, \rho, \eps$)  \COMMENT{ $\pi_v \in \Delta[d]$ }
    \STATE draw $i \sim \pi_v$
    \STATE set $v$.coordinate = $i$
    \STATE set $v$.left\_child = BuildTree($P^v_{i,0}$, $\rho$,$C$, $\eps$)  
    \STATE set $v$.right\_child = BuildTree($P^v_{i,1}$, $\rho$,$C$, $\eps$) 
    \ENDIF
    \STATE Return $v$
\end{algorithmic}
\end{algorithm}

\begin{algorithm}[h!]
    \caption{QueryTree}
    \label{alg:query_tree}
\begin{algorithmic}[1]
    \STATE {\bfseries Input}: query $q \in \mathbb{H}^d$, node $v$, query algorithm parameter $m$, stopping condition $C$
    \STATE $P^v = v$.dataset
    \STATE Select $m$ uniform random points from the current bucket.
    \STATE If one of these points is an approximate near neighbor, then return it.
    \STATE Otherwise,
    \IF{$|P^v| > C$}
    	\IF{$q(v.\text{coordinate})$ = $0$}
        \STATE Return QueryTree(q,v.left\_child, $m$, $C$)
        \ELSE 
        \STATE Return QueryTree(q,v.right\_child, $m$, $C$)
        \ENDIF
    \ELSE
        \IF{approximate near neighbor is in dataset}
        \STATE Return approximate near neighbor 
        \ELSE 
        \STATE Return \O
        \ENDIF
    \ENDIF
\end{algorithmic}
\end{algorithm}

\begin{algorithm}[h!]
\caption{Min-Max Optimization}\label{alg:MinMaxOpt}
\begin{algorithmic}[1]
    \STATE {\bfseries Input}: node $v$, query parameter $\rho$, query parameter $m$, approximation $\eps$
    \STATE initialize weights/distribution $\pi_0 = w_0 = \mathbbm{1}_d \cdot \frac{1}{d}$
    \STATE $T = \frac{10 \ln d}{\epsilon^2}$
    \STATE $\beta = 1 - \sqrt{\frac{\ln d}{T}}$
    \FOR{$t=1,...,T$} 
        \STATE $y_t = \argmin_y \left( \pi_{t-1}^\prime A_v^{\rho} y \right)$ \COMMENT{query player minimization} \label{blackbox}
        \STATE $w_{t+1} = w_{t} \odot \beta^{\ell_{\rho,v} (\pi_{t-1}, y_{t})}$ \COMMENT{hash player update}
        \STATE $\pi_t = \frac{w_t}{\sum_{i=1}^d w_t(d)}$ \COMMENT{normalize weights}
    \ENDFOR
    \STATE Return $\pi_T$
\end{algorithmic}
\end{algorithm}

\subsection{Min-Max Optimization Analysis}
\label{section:minmax}

To solve the min-max optimization, Eqn.~\eqref{eqn:minmaxopt}, efficiently, we iterate a two-player zero-sum game (Def.~\ref{def:minmaxgame}). In this game, the "hash" player selects a distribution over coordinates to hash the dataset on, and the "query" player selects a query/nearest neighbor pair adversarially for the least probability of success at the end of the tree. Using this method, we can find an approximate solution to the min-max program in the following runtime.

\begin{theorem} [Solving the Min Max Optimization]
\label{thm:GameConvFinal}
For any desired $\eps>0$, there exists an algorithm (Algorithm \ref{alg:MinMaxOpt}) that solves the min-max optimization in Eqn. \eqref{eqn:minmaxopt} for the node $v$, up to an additive approximation $\eps > 0$ in $O(\frac{1}{\eps^2} n_v d \ln^2 d )$ time. 
\end{theorem}

The algorithm we describe for this problem exploits results for two-player games. To understand the theorem, we introduce some relevant notions from game theory.

\begin{definition}
\label{def:twoplayergame}
A (simultaneous) \textbf{two-player game} is when two actors (players) are each able to play a weighted mixture of actions (as in Definition \ref{def:purestrat}), without knowledge of the other players mixture, where each action incurs a reward that is a function of the mixtures of both players. The game is characterized by two reward matrices $R,C$ (one for each player) whose entries are indexed by pairs of single actions. The reward for each player is a function of these matrices (as in Definition \ref{def:reward}). This game is called \textbf{iterated} if the game is repeated in sequential rounds. 
\end{definition}

\begin{definition}
\label{def:purestrat}
Suppose a player in a two-player game has $N$ actions available to them. One such action is called a \textbf{pure strategy}, and is represented by a standard basis vector $e_i$ for $i\in[N]$. Further, a \textbf{mixed strategy} $s \in [0,1]^N$ is a convex combination of these pure strategies.
\end{definition}

\begin{definition}
\label{def:reward}
Suppose the first player plays a mixed strategy $x \in [0,1]^N$, and the second player plays a mixed strategy $y \in [0,1]^M$. The \textbf{reward} or \textbf{payoff} for the first player (whose reward matrix is $R$) is $x^\prime R y$, and for the second player (whose reward matrix is $C$) it is $x^\prime C y$. We call the first player, whose strategy left-multiplies their reward matrix, the \textbf{row player}, while the second player, whose strategy right-multiplies their reward matrix, is the \textbf{column player}.
\end{definition}

\begin{definition}
\citep{constAGT}
Consider a two player game where the row player has $N$ possible pure strategies, and the column player has $M$ possible pure strategies. Suppose that the row player has reward matrix $R \in \mathbb{R}^{N \times M}$, and the column player has reward matrix $C \in \mathbb{R}^{N \times M}$. (A two player game is called \textbf{zero-sum} when $R=-C$). Then, a pair of mixed strategies $(x_0,y_0)$ for $x_0 \in \mathbb{R}^N, y_0 \in \mathbb{R}^M$ is considered an \textbf{$\eps$-approximate Nash equilibrium} if and only if the following two conditions hold:
\begin{enumerate}
    \item $x_0^\prime R y_0 \geq \max_x x^\prime R y_0 - \epsilon$,
    \item $x_0^\prime C y_0 \geq \max_y x_0^\prime C y - \epsilon$,
\end{enumerate}
where $x,y$ are taken from the convex hull of available strategies to each player.
\end{definition}

\begin{definition}
Suppose we are performing min-max optimization at node $v$ in an LSH tree with a given exponent $\rho$. We define the matrix $A^{\rho}_v$ to be the \textbf{payoff matrix} for that node. The entries of this matrix $A^{\rho}_{ij,v}$ correspond to a query/near-neighbor pair $(q,p^*)$ (indexed $j$) and dimension $i$. These entries in particular are: $A^{\rho}_{ij,v} := |P^v_{i,q_i}|^{-\rho} \cdot \mathbbm{1}\{q_i = p^*_i\}$. Note that this matrix is exponentially large, and so is never written explicitly.
\end{definition}

\begin{definition}
\label{def:minmaxgame}
The \textbf{hash/query zero sum game} is a two-player zero sum game at a given node $v$ with exponent $\rho$. In this game, the hash player has reward matrix $R = A_v^\rho$ and the query player has reward matrix $C = -A_v^\rho$. In this case, the hash player has $N = d$ possible pure strategies (coordinates to hash on), while the query player has $M = n {d \choose r}$ many pure strategies, as this is the number of possible query/near-neighbor pairs. 
\end{definition}

For our problem, the hash and query players iterate the above two-player zero-sum game. By the celebrated min-max theorem of Nash, there exists a pair of mixed strategies for the hash and query players (i.e. distributions over pure strategies) in the aforementioned game for which no player can improve their reward by deviating from them (a Nash equilibrium) \citep{nash_1950} . To reach this equilibrium, the hash player selects strategies according to the multiplicative weights update rule with the subsequently defined loss function.

\begin{definition}
\label{def:MWU}
Suppose a player in some game has available to them $N$ pure strategies. Fix some parameter $\beta \in (0,1)$. The \textbf{multiplicative weights update (MWU)} method is a method for choosing a mixed strategy over these $N$ possible actions so as to minimize one's loss on a sequence of loss vectors. In particular, suppose a player suffers a sequence of losses $\ell_s(x)$ for $s=1,...,T$. Let $\pi_s$ be their distribution over strategies at round $s$. For the MWU update rule, the player initializes a set of weights to $w_{i,1} = \frac{1}{N}$ for all $i \in [N]$ at round $1$. In subsequent rounds $t > 1$, the player updates these weights according to $w_{i,t+1} = w_{i,t} \odot \beta^{M-R_t(i)}$. Ultimately, the probability of sampling the strategy with index $i$ at round $t$ is $\pi_s(i) = \frac{w_{i,t}}{\sum_{j\in[N]} w_{j,t}}$.
\end{definition}

\begin{definition}
\label{def:losses}
For node $v$ in the LSH tree, $\rho \in (0,1)$, distribution $\pi \in \Delta[d]$, query/NN pair $y=(q,p^*)$ indexed by $j$, and $i \in [d]$, the loss vector for the hash player in a round of game \ref{def:minmaxgame}, $\ell_{\rho,v}(\pi,y) \in [0,1]^d$, has entries:
\begin{align*}
    \ell_{\rho,v}(\pi,y)_i = 1 -  A^{\rho}_{ij,v}
\end{align*}
\end{definition}

Recall that the query player selects the single query/NN pair with the least probability of success on the most recent hash distribution. This can be thought of as an example of the so-called "Follow-the-Leader" (FTL) strategy selection (see \cite{kalai_vempala_2005}). Notably, although FTL strategies on their own do not guarantee convergence to a Nash equilibrium, the query player may implement FTL (as in Definition \ref{def:QPmini}) to achieve convergence, exactly because the hash player uses MWU.

\begin{definition}
\label{def:QPmini}
Let the payoff matrix be $A^\rho_v$, $Q$ the set of possible query/near-neighbor pairs $y$, and $\ell^q_t(y)$ the loss functions at round $t$ of the game. The following equation is defined as the \textbf{query player minimization} (which is an instance of a best-response oracle):
\begin{align*}
    \argmax_y   \ell^q_{t} (y)  = \argmin_y   \pi_{t}^\prime A^\rho_v y 
\end{align*}
\end{definition}

\begin{theorem}[\cite{FreundSchapire}]
\label{thm:GameConv}
Consider the the hash/query zero sum game  (\ref{def:minmaxgame}). Suppose the hash player uses MWU to select strategies with losses as in Definition \ref{def:losses}. Suppose the query player plays its best-response as in the query player minimization (Definition \ref{def:QPmini}). Let $M = n {d \choose r}$ be the total number of possible query/NN pairs to the given dataset (recall this is super-polynomial in dimension). Suppose $T$ rounds of this iterated game have been executed, and let $x_1,...,x_T \in [0,1]^d$ and $y_1,...,y_T \in \{e_i\}_{i=1}^M$ be the mixed row (hash) and pure column (query) player strategies from these rounds, respectively. Then, for a universal constant $K>0$, the pair of strategies $\left(\frac{1}{T} \sum_{t=1}^T x_t, \frac{1}{T} \sum_{t=1}^T y_t \right)$ for the hash and query players, respectively, is a $\frac{K \sqrt{\ln{d}} }{\sqrt{T}}$-approximate solution to ${\max_\pi \min_{y} \pi^\prime A^{\rho} y}$ (and Nash equilibrium in game \ref{def:minmaxgame}).

\end{theorem}

Theorem \ref{thm:GameConvFinal} follows from this theorem, and that the query player minimization can be solved in time $O(n_v d \ln d)$ (Alg.~\ref{alg:QPM}, Supplemental Material).

\subsection{Discussion of the Success Probability Guarantee}
\label{section:succprobdiscussion}

For any query/near neighbor pair $q, p^*\in\mathbb{H}^d$, Theorem \ref{thm:CorrectnessRuntime} requires a parameter $\rho$ that satisfies: ${\mathbb{E}_{i \sim \pi_t}[\mathbbm{1}\{p^*_i = q_i \} \cdot f_{i,p^*,v}^{-\rho}] \geq 1 \label{eqn:condition}}$ for all nodes $v$ in the tree that contain $q,p^*$ (with fewer than $\frac{1}{m}$ approximate near neighbors of $q$), in which case we can lower bound their success probability by $n^{-\rho}$. The second inequality in the theorem states that this $\rho$ can always upper bounded by $\gamma/c \approx 1/c$ (the upper bound for theoretical LSH). A practicioner may interpret this exponent in the following way: provided that your parameter choice $\rho$ is an upper bound for the least possible $\rho$ such that this condition \eqref{eqn:condition} holds, then you are guaranteed $n^{-\rho}$ performance. Further, as the practicioner also may choose $c$ (as in the $(c,r)$-ANN problem), they may tune this $\rho$ aggressively to achieve maximal improvement, and then set $c = \frac{1}{\rho}$ to obtain worst-case guarantees. 

We highlight an important note regarding the dimensions dependence of the algorithm that appears in Section \ref{sec:E} of the Supplement. Crucially, although in the worst-case we require $\Omega(\frac{\ln d}{\eps^2})$ rounds to solve the main min-max game, we can halt the optimization with a data-dependent approximation guarantee. 

\section{Improvement on Datasets Generated from a Mixture Model}
\label{sec:mixture_model}

We now describe a data model in which our algorithm provably performs much better than the standard, optimal LSH \citep{HIM12,OWZ11}. Note this is the only other implementable algorithm for NNS in Hamming space with worst-case guarantees. In particular, recall that the LSH from~\citep{HIM12} simply samples coordinates at random (which would correspond to the LSH Forest with a uniform distribution $\pi$ in each node). To simplify the analysis, we assume the data are in the high-dimensional limit --- specifically where $d \gg \ln(n)$, with $n \gg d$ (e.g. $n = \poly(d)$), and $d \to \infty$. 


We consider a mixture model, where each component has independently chosen (heterogeneous) coordinates.
Specifically, consider a dataset $P$  where each point $x \in P$ is generated randomly such that each coordinate $i \in [d]$ is drawn independently according to $x_i \sim Bernoulli(\eps_i)$, for some fixed $\eps_i \in (0,1)$. This model has been studied before, e.g., in \citep{dubiner2012heterogeneous} (but, for random queries, not worst-case like we do here). There are settings of the $\eps_i$'s where the uniform distribution is still optimal for the independent Bernoulli above (e.g. $\eps_i = 1/2$ for all $i\in [d]$). 

To maximally simplify the model, we consider the case where the coordinates $[d]$ can be partitioned into two sets $S_1, S_2 \subset [d]$ that are $\epsilon_i$-balanced, for $0 < \epsilon_1 < \epsilon_2 \leq \frac{1}{2}$ respectively. In particular, $p_j \sim Bernoulli(\eps_i)$, if $j\in S_i$, for each $p \in P$. Further, we assume the cardinalities of these sets satisfy $|S_i| \gg k$, where $k$ is the number of hashes chosen by the algorithm (tree depth). Note that although we analyze the case of two such sets, the argument generalizes to many sets (at least a constant number of sets with respect to dimension). The sizes of these sets change as hashing is performed, so we denote these sets relative to a node $v$ in the LSH tree by $S^v_i$. 

Finally, our model is defined simply as a mixture of two clusters each from (essentially) a heterogeneous-coordinates distribution as above. In particular, the second cluster is obtained by planting a point $p_a = 0^d$ and $\sqrt{d}$ points next to the point $p_a$. These points are generated i.i.d. at distance $r+1$ from $p_a$, where the coordinates on which they each differ from $p_a$ are all in $S_1$. Note that in the high-dimensional limit, these additional points will not affect the balances of the coordinates for subsets larger than $d$ (as these planted points compose at most $O(1/\sqrt{d}) \to 0$ fraction of the bucket). 

We show that, in such a model, our algorithm obtains improved performance over uniform hashing: see informal Theorem \ref{thm:informal_improvement} below. The formal statements/proofs for standard LSH \citep{IndMot98} in Theorem \ref{thm:improveUni} and for LSH Forests \citep{bawa2005lsh} in Theorem \ref{thm:improveUni2}.  We note that, interestingly, in the simpler setting of just one cluster, uniform hashing remains essentially optimal (Theorem \ref{thm:uniOptimal}).

\begin{theorem}[Informal]
\label{thm:informal_improvement}
In the above mixture model, trees constructed and queried with Algorithms \ref{alg:build_tree} and \ref{alg:query_tree} obtains a factor of $\Omega \left(\exp(\Omega(\sqrt{\ln d})) \right)$ improvement on the minimum query over LSH forests \citep{bawa2005lsh} and standard LSH \citep{IndMot98} with exponent parameters $\rho \in (0.1,1)$ and $\rho \in (0.2,0.8)$, respectively, and query parameter $m = 0$.
\end{theorem}

Our algorithm obtains improvement over uniform hashing because the optimized distributions in this setting place more weight on the more balanced coordinates (where the Bernoulli parameter is closer to $1/2$). By design, the worst-case query in this data model is the query with bits flipped on only the coordinates that differentiate the planted $p_a$ from its approximate near neighbors. Therefore, placing more weight on the balanced coordinates quickly separates points in the "hard" cluster from the "easy" cluster, as compared to uniform hashing. 

\section{Experiments}
\label{sec:experiments}

We demonstrate the practicality and performance of our algorithm on the canonical ImageNet and MNIST datasets. In this section, we display results for the first 750 images of MNIST’s training dataset \citep{MNIST_web}, and on the first 624 images of ImageNet's 3x8x8 validation subset \citep{ImageNet}. We performed additional experiments on the entire MNIST dataset and a 100,000-point subset of ImageNet's training set, which can be found in section \ref{sec:AdditionalExperiments}. We note that we expect the improvement to be more substantial with larger datasets with a scaled-up algorithm. This is because LSH-type algorithms have success probability/query time of the form $n^{\rho}$, and our experiments already show that our algorithm obtains an improved exponent $\rho$. More specifically, small experiments allowed for the minimum success probability to be greater than $\frac{1}{100}$. In this case, only roughly $100$ trees were needed to resolve this minimum. 

For both MNIST and ImageNet, the dataset was binarized using a threshold. In particular, all pixel values below a threshold pixel value were set to 0, and the complement is set to 1 (a threshold of 16 for ImageNet, and 1 for MNIST). The implementation details can be found in Section \ref{sec:implementationdeets}, Supplemental Material. For the small subsets, we ran our algorithm with radius $r=5$ for ImageNet and MNIST. Two additional parameters are listed for the experiments - the number of rounds $T$ the game was executed for, and the base $\beta \in (0,1)$ used for MWU. 

We compare trees generated by our algorithm to LSH forests (uniformly sampling coordinates). The algorithms with best (average-case) empirical performance on specific datasets for nearest neighbor search have no guarantees (correctness, or performance). Due to this, when measured with respect to our property of interest – minimal success probability over all queries – all such algorithms without theoretical guarantees collapse, i.e., achieve 0 success probability on the worst query. Our goal is different: we want the best algorithm among those guaranteed to do well on all possible datasets and queries. Therefore, we compare our algorithm only to those that are both implementable and have guarantees for the worst query. This is just uniform LSH.

To assess the performance of our algorithm in these settings, for MNIST, we sample 100 points uniformly at distance $r=10$ from each point in the dataset. For ImageNet, we sampled 2 points at distance $r=5$ from each point, and computed success probability similarly. We sample $110$ trees for a range of parameters, and estimate the probability of success for each query/NN pair by computing the fraction of trees which co-locate the pair in their final bucket. In these experiments, we do not sample pivots as in Algorithm \ref{alg:query_tree} to more directly compare the quality of the optimized hash functions to uniform ones.

The experiments show that our algorithm with certain parameters produces trees with a $1.8\times$ improvement over uniform hash trees in the success probability for the minimum query for MNIST (Table \ref{table:MNIST_succ_probs} and Figure \ref{fig:MNIST_histo}), and $2.1\times$ improvement for the bottom tenth percentile of queries to ImageNet (Table \ref{table:ImageNet_succ_probs}). These success probability improvements are accompanied by large query time improvements for both datasets (Table \ref{table:ImageNet_times}).

One might ask - what kinds of distributions will our optimization produce in practice to obtain this improvement? To answer this question, we show the distributions produced by the min-max optimization in Algorithm \ref{alg:MinMaxOpt} at the root of the MNIST dataset for two settings of the exponent parameter $\rho$ (see Figure \ref{fig:distributions}). For MNIST, the distributions that produced the greatest improvement over uniform hashing placed more weight on pixels towards the center of the image (and significantly less weight in the corners). A factor that contributes to this phenomenon is that coordinates closer to the center of the image are much more balanced, and hence are favored by the optimization.

\begin{table}[H]
  \caption{Success Probability on Random Queries to a subset of ImageNet.}
  \label{sample-table}
  \vskip 0.15in
  \centering
  \begin{tabular}{lll}
    \toprule
    Parameters & Bottom 10\%  & Average \\
    \midrule
    Uniform & 0.275 & 0.621 \\
    $\rho=1$, $T=3000$, $\beta=0.68$ & \textbf{0.576} & \textbf{0.772}\\
    \bottomrule
  \end{tabular}
\label{table:ImageNet_succ_probs}
\end{table}

\begin{table}[H]
  \caption{Success Probability on Random Queries to a subset of MNIST.}
  \label{sample-table}
  \vskip 0.15in
  \centering
  \begin{tabular}{lll}
    \toprule
    Parameters & Minimum & Average \\
    \midrule
    Uniform & 0.35 & 0.737\\
    $\rho=1$, $T=3000$, $\beta=0.68$ & 0.6 & 0.877 \\
    $\rho=0.83$, $T=3000$, $\beta=0.68$ & \textbf{0.63} & \textbf{0.878} \\
    $\rho=0.25$, $T=1600$, $\beta=0.88$ & 0.42 & 0.834 \\
    $\rho=0.1$,  $T=1600$, $\beta=0.88$ & 0.36 & 0.785 \\
    \bottomrule
  \end{tabular}
\label{table:MNIST_succ_probs}
\end{table}

\begin{figure}[H]
\begin{subfigure}{0.5\textwidth}
  \centering
  \includegraphics[scale=0.18]{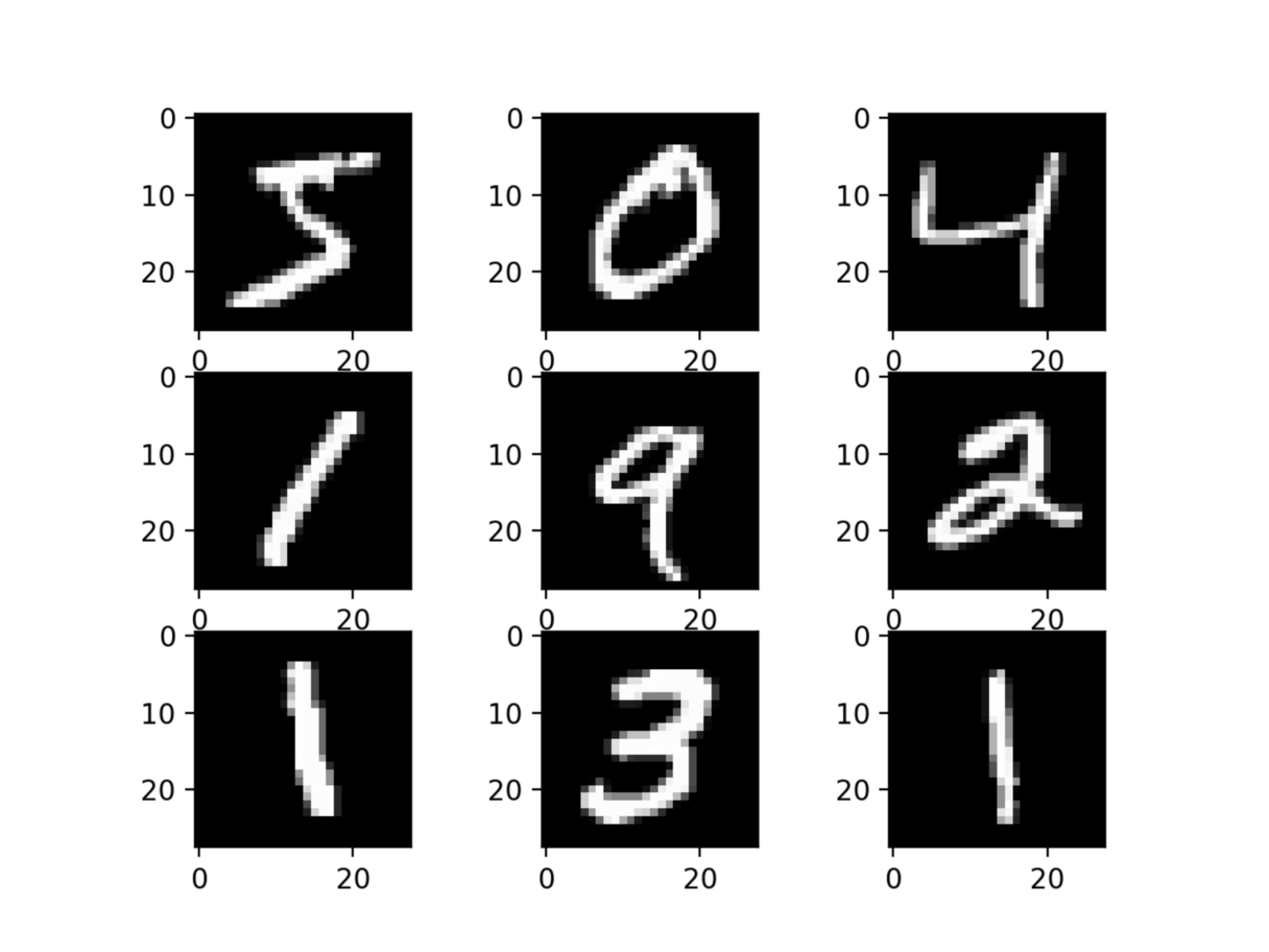}
  \caption{Example MNIST Digits}
  \label{fig:balance}
\end{subfigure}
\begin{subfigure}{0.5\textwidth}
  \centering
  \includegraphics[scale=1.35]{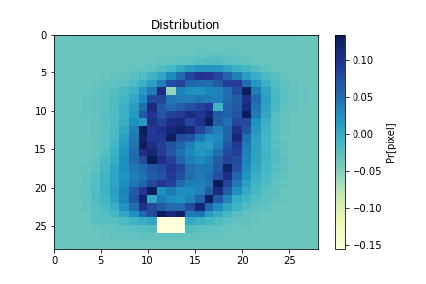}  
  \caption{$\rho = \frac{5}{6}$, $T=3000$, $\beta=0.68$}
  \label{fig:dist_c1o2}
\end{subfigure}
\begin{subfigure}{0.5\textwidth}
  \centering
  \includegraphics[scale=1.35]{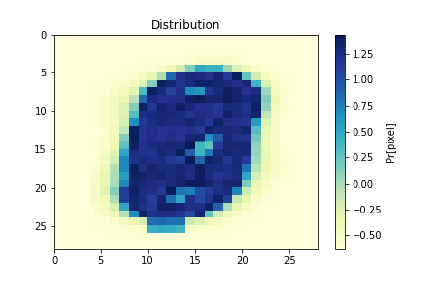}  
  \caption{$\rho = \frac{1}{4}$, $ T=1600, \beta=0.88$}
  \label{fig:dist_computed}
\end{subfigure}
\caption{Scaled and centered distributions produced by Algorithm \ref{alg:MinMaxOpt} for the MNIST dataset (optimized for the entire dataset)}
\label{fig:distributions}
\end{figure}

\begin{figure}[H]
\begin{center}
\begin{subfigure}{.4\textwidth}
  \centering
  \includegraphics[scale=1.4]{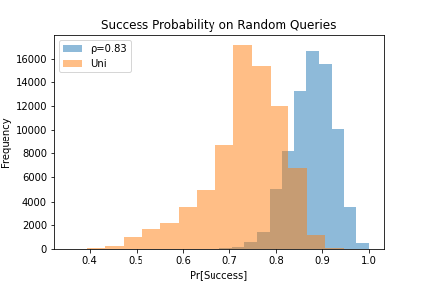}  
\end{subfigure}
\caption{Histogram of Recall for 75,000 Random Queries to MNIST. Algorithm parameter settings $\rho = \frac{5}{6}$, $T=3000$, $\beta=0.68$.}
\label{fig:MNIST_histo}
\end{center}
\end{figure}

\section{Discussion}

\paragraph{Challenges for designing instance optimal NNS algorithms.} An ideal goal for data-aware NNS would be an instance-optimal algorithm: one that achieves the best possible performance among all possible algorithms. To avoid hard computational complexity issues, it is only reasonable to ask for all possible algorithms from a restricted class of algorithms $\cal C$, for as large class $\cal C$ as possible. 

We considered the class $\cal C$ of, essentially, (random) decisions trees, where each node is a coordinate cut (in the Hamming space). Our algorithms is instance optimal as long as the algorithm knows the correct parameters $\rho$ at each node. 

It would be natural to try to extend the class $\cal C$ to include other possible hashes (node decision functions), most notably hyperplane cuts for the Euclidean space, and ball cuts (for both Hamming and Euclidean spaces). Such hashes are popular for practical and theoretical LSH algorithms. 

There are some challenges in extending our algorithm to the above settings. Specifically, while one can efficiently extend the algorithm in this paper to other hash functions and metrics, the runtime must depend {\em polynomially either on the number of possible hash functions or  the number of possible queries}. 
Indeed, one can solve the two-player game by implementing MWU strategy selection for the player with polynomially many strategies (in the dataset size), and FTL for the other player. Alas, both for hyperplane and ball cuts the number of hash functions and queries are essentially exponential in $d$. It may be possible to reduce the number of hash functions/queries by making an assumption: e.g., consider ball cuts with centers at dataset points, or assume queries come from a distribution.

\paragraph{Euclidean distance case.} While we focus on Hamming distance in this work, it is possible to extend our algorithm to Euclidean space. In particular, it is a classic result that one can embed Euclidean space $\ell_2$ into $\ell_1$ and hence Hamming space, up to appoximation arbitrarily close to 1~\cite{figiel1977dimension}. In particular, one can observe that the resulting algorithm would correspond to picking $\Theta(d)$ random hyperplanes and then optimizing only with respect to them. This can be seen as another approach to optimize over large classes of hash functions: not optimize with respect {\em all} hashes, but with respect to only $\Theta(d)$ (or perhaps $\poly(d)$) randomly-chosen hashes. We leave this direction of exploration for future research.

\paragraph{Pre-processing.}

Recall that the runtime for pre-processing of our algorithm is $n \cdot \poly(d) \cdot \frac{1}{\eps^2}$, where $\eps$ is the approximation factor in the min-max game. To closely approximate the optimum success probability, we need to set $\eps$ to be on the order of (ideally less than) the optimum. Therefore, an equivalent runtime is $n \cdot \poly(d) \cdot (\Pr[success])^{-2}=n^{1+O(\rho)}\cdot \poly(d)$.

We also note that, for our algorithm, one can tradeoff query time improvement for faster pre-processing, without sacrificing the worst-case guarantee. In particular, one can optimize on any subset of nodes and hash uniformly otherwise (e.g. only optimize on the final levels of the trees), while retaining the lower bound in Theorem \ref{thm:CorrectnessRuntime}. We obtained improvement over uniform hashing with this approach for datasets with $\approx 10^5$ points (see Section \ref{sec:AdditionalExperiments}).

\section*{Acknowledgements}

We acknowledge support from the Simons Foundation (\#491119 to Alexandr Andoni), and NSF  grants  CCF-2008733, CCF-1617955, and CCF-1740833. 

\bibliographystyle{authordate1}
\bibliography{ref/BIB,ref/andoni}

\newpage
\appendix

\section*{Supplement to "Learning to Hash Robustly, with Guarantees"}

\subsection*{Preface}
We briefly outline the structure of the supplement.  We highlight the Further Discussion in Section \ref{sec:E}, which includes, among other discussion points, the important note which describes how to obtain an approximation guarantee for the main min-max game that is specific to the given dataset. In particular, a practitioner can halt the game in orders of magnitude fewer iterations than the worst-case and retain this optimization guarantee.

In Section \ref{sec:B}, we describe an algorithm for efficiently implementing the query player minimization. In Section \ref{sec:C}, we derive the main correctness and performance theorem. In Section \ref{sec:D}, we describe a data model and prove for that model that our algorithm can perform much better than uniform hash functions on the worst-case queries. In Section \ref{sec:UniOptimal}, we demonstrate that for LSH forests and a reasonable LSH variant, the uniform distribution is optimal for a single cluster with independent coordinates. In Section \ref{sec:AdditionalExperiments}, we demonstrate the practicality and performance of our algorithm with a variety of additional experiments. In Section \ref{sec:F}, we include hardware details for our experiments and the link to the code used to generate our empirical results. 

\newpage

\twocolumn
\section{Further Discussion}
\label{apx:discussion}
\label{sec:E}

\paragraph{Convergence of the min-max game.} The number of steps to convergence depends as $1/\eps^2$ on the error $\eps$ to success probability. As noted above, $\eps$ must be of the order $\Pr[success]$ hence quite small, and it is normal to wonder whether this step can be sped-up. As we discuss below, the game can be stopped much earlier, in a data-aware way, while retaining the theoretical guarantees.

Indeed, in our experiments, we used far fewer iterations than are required by Theorem \ref{thm:GameConvFinal}. We used $3000$ iterations, while the theoretical bound requires at least $\frac{64 \ln(784)}{\eps^2} \geq 42,000$ iterations to achieve a $0.1$-approximation to the optimum. How can we show that the distributions are converged with much fewer iterations than the worst-case bound? In general, how might a practitioner improve our algorithm's polynomial dependence on the dimension? In the proof of Theorem \ref{thm:GameConv} from \citep{FreundSchapire}, the proximity to a Nash equilibrium is bounded by the regret of hash player (where regret is defined below).

\begin{definition}
Suppose a player in a two-player zero-sum game has played a sequence of mixed strategies $x_1, ..., x_{T}$ up to time $T$, each of which has experienced some loss according to the functions $\ell_s(x)$ for $s \in 1,...,T$. Then, that player's regret is defined as follows.
\begin{align}
    Regret(x_1,...,x_{T}) := \sum_{s=1}^{T} \ell_s(x_s) - \min_x \sum_{s=1}^{T} \ell_s(x) 
\end{align}
\end{definition}

Therefore, a practitioner can simply halt the game when the regret of the hash player is less than their desired approximation threshold. 

We demonstrate this data-dependent bound for our experimental setup. Consider the distribution at the root of a 700-point ImageNet subset for an experiment with $\rho=1$, $T=3000$, and $\beta = 0.4$. For the hash player, we compute the best distribution in hindsight by solving the following linear program:
\begin{align}
    &\min_\pi \frac{1}{T} \sum_{t=1}^T (1 - \Pr_\pi[\text{success on $(p^t,q^t)$}]) \\
    &= \min_\pi \left( 1 - \sum_{t=1}^T \sum_{i=1}^d \pi_i \mathbbm{1}\{p^t_i = q_i\} \cdot n_{i,t}^{-\rho} \right) \\
    &= 1 - \max_\pi \sum_i \pi_i \left( \sum_t \mathbbm{1}\{p^t_i = q_i\} \cdot n_{i,t}^{-\rho} \right)
\end{align}
Where $\pi$ is constrained to $\Delta[d]$. 

Solving this program gives that the best distribution in hindsight had loss $0.99795$, meaning that the optimal hash strategy had success probability $0.00205$. Meanwhile, the loss incurred by the hash player is $0.99812$. Therefore, the query/hash strategies are within $0.000169$ of the game's value, which in this case is within 10\% of the optimal strategy in $14\times$ fewer iterations than required in the worst-case! 

\paragraph{Instance optimal algorithms with many queries and hashes.}

It is still conceivable to design an efficient algorithm for instance optimal hashing even when there are exponentially many queries and hash functions. Intuitively, say, from the perspective of the hash function distribution, we do not actually need the optimal distribution---merely a sample from it. In fact, if Karlin's weak conjecture holds \citep{karlin_1962}, namely that both players can use FTL to achieve sublinear regret, then neither player must explicitly consider all of their exponentially many strategies! There is some hope that this conjecture is true (see \citep{abernethy_lai_wibisono_2021}). 

\paragraph{Effect of the parameter $\rho$ to the algorithm.} 
Depending on the exponent $\rho$ chosen for the optimizations in the experiments, the distributions returned by the min-max optimization in Algorithm \ref{alg:MinMaxOpt} could be qualitatively quite different. In the experiment with exponent $\rho=\frac{5}{6}$ (Figure \ref{fig:distributions}), the distribution at the root placed a large amount of weight on the most balanced bits. However, when the exponent was decreased to $\rho=\frac{1}{10}$, the optimal distribution (roughly) uniformly weighted many balanced and unbalanced bits.

We illustrate why the optimal distributions might be very different depending on input $\rho$ with the following examples. Suppose there are two groups of bits with balances $0.1$,$0.5$, respectively, and with sizes $5r$,$r$ for queries at radius $r$. Then, the optimal distribution will have less weight on the more balanced group, as the worst-case query with have all $r$ of the more balanced bits flipped. Suppose instead that the two groups have balances $0$,$0.5$, with sizes $r$,$100r$, respectively. Then, the optimal distribution will place no weight on the first, unbalanced group, as these bits make no progress on hashing the dataset, while the second group is sufficiently large that you can increase the weight on that group without increasing the probability of failure substantially. 

Theorems \ref{thm:improveUni} and \ref{thm:improveUni2} show large improvement over uniform hashing by exploiting the nice structure of the dataset to isolate the cluster later down the tree. There we make an assumption that the coordinates differentiating the cluster all appear in the less-balanced group. We expect that, in practice, with sufficiently "diverse" data, a given dataset might consist of many clusters whose differentiating coordinates are spread across balances. (Note we still see theoretical improvement in this case.) The function of these clusters is really to introduce adversarial quality to the dataset so that the current balances at the root (or a nearby descendant) do not indicate the true difficulty of hashing on a particular coordinate. In this case, the worst-case queries are not those with the most balanced coordinates flipped, but rather those embedded in these adversarial clusters. In effect, the worst-case analysis of the LSH tree is reduced to average-case analysis (as the differentiating coordinates are likely spread across a large spectrum of balances). 

Recall that with the "correct" objective values our algorithm is exactly instance-optimal. Given the discussion above, this suggests that there might be a different choice of objective values in our min-max game that more closely approximates instance-optimality. In particular, rather than setting a single exponent $\rho$ for the entire algorithm, this parameter should vary for each possible bucket, and should be tuned reflect the difficulty of hashing uniformly on that bucket. In other words, while our algorithm optimizes for the snapshot of a dataset at a given node, an instance optimal algorithm must be able to look ahead and use information about what the buckets will look like in future nodes. 

As we showed in section \ref{sec:UniOptimal}, if the coordinate balances remain constant throughout the hash tree (in which case the bucket looks identical to the optimizer at every node), it may not be possible to improve over uniform at all! Therefore, to guarantee improvement over uniform sampling, we need that the balance profiles of the coordinates change throughout the hash tree. This is likely what occurred in our experiments - namely, which coordinates were balanced early in the tree were distinct from (or independent of) which coordinates were balanced later down the tree. To exploit this feature of a dataset in practice, one may try to obtain a better lower bound on the probability of success than merely $n^{-\rho}$, for $\rho$ from the uniform LSH---e.g., by setting up a convex program and relaxing it. This is an interesting route for future work.

\section{Implementing Query Player Minimization}
\label{section:qpm}
\label{sec:B}

The implementation of the MWU strategy selection is fairly self-explanatory, but how should we efficiently perform the query player minimization? We now show that the query player minimization is efficiently implementable and prove the runtime of the entire pre-processing procedure.

\begin{algorithm}
\caption{Query Player Minimization} \label{alg:QPM}
\begin{algorithmic}[1]
    \STATE {\bfseries Input:} $\pi_v$, $v$, $\rho$
    \STATE {\bfseries Output:}  (min query)
    
    \STATE {\it \{$\pi_v$ - distribution over coordinates for node $v$\}} \\~\\
    \STATE compute objective values $n_{i,b}^{-\rho}$ for $i \in [d]$ and $b \in \{0,1\}$ \\
    \STATE (min probability) = $\infty$
    \STATE (min query) = None
    \FOR{$j=1,...,n_v$}   
        \STATE set $u_j \in \mathbb{R}^d$ to the objective values for the given datapoint.
        \STATE compute $s_j = u_j \odot \pi_v$
        \STATE sort $s_j$, while tracking the positions of the original coordinates
        \STATE Set the top $r$ values in the sorted list to $0$ and sum the remaining values (call this sum $z_j$)
        \IF{$z_j < $ (min probability)}
            \STATE set (min probability) $= z_j$
            \STATE set (min query) to the current datapoint $j$ with the top $r$ coordinates flipped. 
        \ENDIF
        
    \ENDFOR
\end{algorithmic}
\end{algorithm}

\begin{proof}[Proof of Theorem \ref{thm:CorrectnessRuntime} (Pre-processing Time)]
The query player minimization (step \ref{blackbox} of Algorithm \ref{alg:MinMaxOpt}) can be implemented exactly using Algorithm \ref{alg:QPM}.

Suppose the current node is $v$, and the current bucket is of size $n_v$. By theorem \ref{thm:GameConvFinal}, a single min-max optimization can be solved in time $O(\frac{1}{\eps^2} n_v d \ln^2 d )$. In a given layer of the tree, each node contains a bucket that is disjoint from all other buckets in that layer. Therefore, the total runtime for the algorithm on a single layer of the tree is $O(\frac{1}{\eps^2} n d \ln^2 d )$. Further, there are at most $d$ layers in the tree, as we hash without replacement when we progress to the next layer. This gives the pre-processing time in the theorem. 
\end{proof}

\section{Proof of Theorem \ref{thm:CorrectnessRuntime} (Success Probability)}
\label{section:SuccProbGuarantee}
\label{sec:C}
We now prove the success probability guarantee as in Theorem \ref{thm:CorrectnessRuntime}. Let $\delta = \frac{1}{m}$ for chosen tradeoff parameter $m > 0$. In particular, suppose we are given a dataset with a datapoint $p^*$ and a query $q$ with $\|p^* - q\|_1 \leq r$. Recall that we want to guarantee that when the querying procedure terminates (Algorithm \ref{alg:query_tree}), the probability that the pair of points collide on the final bucket is at least $n^{-\rho} \geq n^{-\frac{\gamma}{c}}$, for $\rho$ satisfying $\mathbb{E}_{i \sim \pi_t}[\mathbbm{1}\{p^*_i = q_i \} \cdot f_{i,p^*,v}^{-\rho}] \geq 1$ on all nodes $v$ in the LSH tree (where $f_{i,p^*,v} = |P^v_{i,p^*}| / |P^v|$). 

\begin{proof}
The proof is by induction over the size of the dataset. Fix any query/NN pair $q,p^* \in \mathbb{H}^d$. For the base case assume the size of the dataset is $|P| = 1$. Then, by assumption that there is a near neighbor in the dataset, and as the stopping condition is reached ($1 \leq C$ for all choices of $C$), the probability of success is exactly $1$.

We now prove the induction step of the claim. Consider the tree of possible hashes from the given dataset, with each child corresponding to a hash event. Note this is a $d$-ary tree. Consider a dataset of size $n_v$ at some node $v$ in the tree, with some children that have additional optimizations performed and perhaps some children that don't. Suppose all children have size $n_{i,v} < n_v$. If not, we re-direct this argument to the child with $n_{i,v}  = n_v$. If this child also has children of size $n_v$, then we again focus on the grandchild node, repeating this recursion until we reach a descendant node with children all of size strictly less than $n_v$. We can then "unpeel" the argument to prove the inductive hypothesis for the original node using the same calculation as below. 

Assume for induction that the optimizations in the children (one child for each $i \in [d]$) produce a distribution that has minimum probability of success greater than $n_{i,v}^{-\rho}$ for all queries with $n_{i,v} < n_v$.  The current node optimizes assuming the children have probability of success $\{n_{i,v}^{-\rho}\}_i$. We follow a similar approach to \citep{andoni_razenshteyn_nosatzki_2017} to prove lower bounds with uniform hashing. 

For the first inequality, recall by the theorem assumption, there exists a distribution over hashes $\pi_v$ such that $\mathbb{E}_{i \sim \pi_v}[\mathbbm{1}\{p^*_i = q_i \} \cdot f_{i,p^*_i, v}^{-\rho}] \geq 1$. Then, we have the following by the induction hypothesis:

\begin{align}
    \Pr[success]  &= \sum_i \pi_{v,i} \mathbbm{1}\{p^*_i = q_i \} \cdot \Pr[success \mid P^v_i] \\
    &\geq \sum_i \pi_{v,i} \mathbbm{1}\{p^*_i = q_i \} \cdot n_{i,v}^{-\rho} \\
    &= n_v^{-\rho} \cdot \mathbb{E}_{i \sim \pi_v}[\mathbbm{1}\{p^*_i = q_i \} \cdot f_{i,p^*_i, v}^{-\rho}] \\
    &\geq n_v^{-\rho}
\end{align}

\noindent This completes the proof for the first guarantee of the success probability.\\

To ensure the second inequality in the theorem holds (namely $\rho \leq \frac{\gamma}{c}$), we follow the strategy suggested in \citep{andoni_razenshteyn_nosatzki_2017} to handle datasets with many approximate near neighbors. In particular, we select points uniformly at random at each node and compare these to our query, halting the query procedure if an approximate near neighbor is found. Suppose then for the node $v$ in the LSH tree, the current bucket $P^v$ has $\geq \delta n_v$ points at distance less than $cr$ from the query. Then, with constant probability, selecting $m = \Theta(\frac{1}{\delta})$ (as the algorithm parameter) random points in the dataset will include one such near point. Note that this is why the query time in Theorem \ref{thm:CorrectnessRuntime} is $m d^2$, as a single query comparison takes $d$ time and there are at most $O(md)$ comparisons at query time. Now suppose to the contrary that fewer than $(1-\delta)n_v$ points are at  distance less than $cr$ from all queries. 

For the second inequality in the theorem ($\rho \leq \frac{\gamma}{c}$), we note for randomized hash function $h$ with distribution $\pi \in \Delta[d]$, where we let $\rho_0 = \frac{\gamma}{c}$:
\begin{align}
\Pr[\text{success}]  &= \Pr_{i\sim \pi} [\text{success on $P^v_{i,q_i}$} ] \\
&\geq \Pr_{i\sim \pi} [\text{success on $P^v_{i,q_i}$, and $p^*_i = q_i$}] \\ 
&= \Pr [\text{success on $P^v_{i,q_i}$} \mid p^*_i = q_i ] \cdot \Pr_\pi [p^*_i = q_i ] \\
&\geq  \mathbb{E} [n_{i,v}^{-\rho_0}]\cdot \Pr_{i\sim \pi} [p^*_i = q_i ],
\end{align}
where the fourth step follows by the induction assumption. If we choose $h$ to be distributed uniformly, then applying Jensen's inequality we get:

\begin{align}
\Pr[&\text{success}]  \geq \left(1 - \frac{r}{d}\right) \cdot \mathbb{E} [n_{i,v}^{-\rho_0}] \\
&\ge \left(1 - \frac{r}{d}\right) \left( \sum_{p \in P^v} \Pr[p^*_i = p_i ] \right)^{-\rho_0} \\
&\ge\left(1 - \frac{r}{d}\right) \left(\frac{\delta n_v}{n_v} + \frac{(1-\delta) n_v}{n_v}(1 - \frac{cr}{d}) \right)^{-\rho_0} \cdot n_v^{-\rho_0}
\end{align}

The third line is because we assume at most $\delta$ fraction of points are at distance at most $cr$, and at least $1-\delta$ fraction are at distance at least $cr$. To complete the proof, we now show the last formula is lower bounded by $n_v^{-\rho_0}$ .
\begin{align}
    (1 - \frac{r}{d}) \left(\frac{\delta n_v}{n_v} + \frac{(1-\delta)n_v}{n_v}(1 - \frac{cr}{d}) \right)^{-\rho_0} &\geq 1 \\
    \impliedby \left(\frac{\delta n_v}{n_v} + \frac{(1-\delta)n_v}{n_v}(1 - \frac{cr}{d})\right)^{-\rho_0} \geq \frac{1}{1 - \frac{r}{d}}& \\
    \impliedby \rho_0 \geq \ln(\frac{1}{1 - \frac{r}{d}}) \ln^{-1} \left( \frac{1}{\delta + (1-\delta)(1 - \frac{cr}{d}) }\right)& \\
    \impliedby \rho_0 \geq \ln(\frac{1}{1 - \frac{r}{d}}) \ln^{-1} \left( \frac{1}{1 - (1 - \delta) \frac{cr}{d} } \right)& \label{eqn:precise_rho}
\end{align}

Note that $\frac{1}{(1 - \delta)c} \geq \ln(\frac{1}{1 - \frac{r}{d}}) \ln^{-1} \left( \frac{1}{1 - (1 - \delta) \frac{cr}{d} } \right)$, and so we can set $\rho_0 \geq \frac{1}{(1 - \delta)c}$. As the true probabilities of success are greater than the "lower bound" objective by the induction assumption in both inequalities, the true probability of success is greater still than $n_v^{-\rho}$ (or $n_v^{-\rho_0}$ in the second inequality), proving the theorem.

\end{proof}

\section{Formal Treatment of Mixture Model}
\label{apx:improveUniProof}
\label{sec:D}

In the theorems that follow, we consider the mixture model where there are two disjoint sets of coordinates of equal size with homogenous balances. While we assign specific attributes to these sets, the proof method can apply to general instances of this mixture model with many sets of different balances and sizes - naïvely, by taking weighted averages of these quantities. Thus, we do not require these specific parameter settings to show improvement. 

We first define the uniform LSH algorithm \citep{IndMot98} for the general ANN problem. In this algorithm, for a chosen approximation factor $c$, a fixed number of hash functions are chosen such that the probability of success for the algorithm, for any query at distance $r$ from its near neighbor, is exactly $n^{-\rho}$ where $\rho = \frac{\ln(1 - \frac{r}{d})}{\ln(1 - \frac{cr}{d})}$. 

\begin{theorem}
Suppose we are given a dataset drawn according to the above data model, with ${r = d/\sqrt{\ln d}}$, $n = d^6$, $\eps_1 = 0.3,\eps_2 = 0.5$, $|S_1| = |S_2| = \frac{d}{2}$. Then with probability $0.99$ over the data distribution, trees constructed and queried with Algorithms \ref{alg:build_tree} and \ref{alg:query_tree} with algorithm query parameter $m$, exponent parameter $\rho \in (0.1,1)$ until the bucket size is $d$, and after that $\rho=0$ until stopping condition $C=1$, where $k_u\triangleq\ln(\frac{n}{d})\ln^{-1}(1/\eps_1)$ and ${ \bar{\eps} \triangleq  \frac{\ln(1 - \frac{r}{d})}{\ln(1 - (1-\frac{1}{m})\frac{r+1}{d})} - \frac{\ln(1 - \frac{r}{d})}{\ln(1 - \frac{r+1}{d})}}$, has success probability at least ${d^{-\bar{\eps} } (1 - \frac{1}{\sqrt{\ln(d)}})^{-k_u} = \Omega\left(d^{-\bar{\eps}} \exp(\Omega(\sqrt{\ln d})) \right)}$ times greater than uniform LSH on the minimum query for ANN with approximation factor $c=1+\frac{1}{r}$.
\label{thm:improveUni}
\end{theorem}

\begin{proof}[Proof of Theorem \ref{thm:improveUni}]

We must first understand what it means for a query to be "worst-case" for the standard uniform LSH. In particular, this algorithm in its original formulation uses a fixed number of uniform hash functions, and so the probability of success is the same for all queries at a fixed distance from their near neighbor. To define the probability of success for uniform LSH as applied to our data model, we divide the potential queries to this dataset into two classes. In the first class, we consider queries to any arbitrary point (not equal to $p_a$ and its cluster), which all require an equal number of hashes to reach expected bucket size $1$. In the second, we consider queries to $p_a$ with bits flipped on the coordinates that differentiate $p_a$ from its planted approximate near neighbors. For the second class to be "worst-case" we need that the probability of success for queries in this class are less than the first. 

The probability of success for the second class is exactly $\sqrt{d}^{-\rho_u}$ where $\rho_u = \frac{\ln(1 - \frac{r}{d})}{\ln(1 - \frac{r+1}{d})}$. We can lower bound this success probability by $\sqrt{d}^{-\frac{1}{1 + O(\frac{1}{r})}} \ge \frac{1}{d^{1/2+o(1)}}$.

The probability of success for the first class is lower bounded by $n^{-0.5/c_0}$ where $c_0 r$ is the average distance between two points (chosen iid from the model). We can compute this distance as $c_0 r = \frac{d}{2}\cdot 2(1-\eps_1)\eps_1 + \frac{d}{2}\cdot 2(1-\eps_2)\eps_2 = d(1-\eps_1)\eps_1 + d(1-\eps_2)\eps_2$. As we have set $n = d^6$, we have that the probability of success for this first class is,
\begin{align}
    \ln \Pr[success \mid \text{for phase $1$}] &\geq -6\frac{0.5}{\frac{0.46d}{r}} \ln d \\
    &> -7 \frac{r}{d} \ln d \\
    &= -7 \sqrt{\ln d}  \\
    \iff \Pr[success \mid \text{for phase $1$}]  &\geq \exp ( -7 \sqrt{\ln d}) \\
    &\gg \exp ( -0.5 \ln d) \\
    &= \frac{1}{\sqrt{d}} \\
   \approx \Pr[success &\mid  \text{for phase $2$}] 
\end{align}

proving that the second class queries are indeed worst-case in the high-dimensional limit.

Because the distribution for the optimized hash functions are maximal for their objective (by definition), we can choose any distribution we'd like and derive a lower bound for the performance of a single optimized hash distribution. We consider distributions that are marginally uniform on each group $S_i$, as the planted point has a 0 on each coordinate (and so the coordinates are symmetric across groups). Suppose the optimized distribution is $\pi = (0,1)$. This is the distribution that would be returned by our algorithm for almost all $\rho$, but certainly including e.g. $\rho \in (0.1,1)$. To see this, we first note that the objective function (the lower bound for the probability of success for $\pi = (\pi_1,\pi_2)$) is:
\begin{align}
    \text{Objective} = \frac{d}{2} \pi_1(1 - \eps_1)^{-\rho} + \frac{d}{2} \pi_2(1 - \eps_2)^{-\rho}
\end{align}
Then, as,
\begin{align}
    (1 - \frac{2r}{d})(1 - \eps_2)^{-\rho} > \frac{1}{2}(1 - \eps_1)^{-\rho} + \frac{1}{2}(1 - \frac{2r}{d})(1 - \eps_2)^{-\rho}
\end{align}
we conclude $\pi = (0,1)$ is the distribution returned by our algorithm.

Suppose we choose $k_u$ hash functions to reach $d/n$ fraction of points remaining in the original dataset. The probability of success for the uniform distribution on the worst-case query on reaching this fraction is ${(1 - \frac{r}{d})^{k_u}}$.

As a uniform hash function reduces the dataset to at least $\eps_1$ fraction of the original dataset size, $k_u \geq \ln(\frac{n}{d}) \ln^{-1}(1/\eps_1)$. 

Meanwhile, for the worst-case query, as we have assumed all of the coordinates that differentiate the cluster center $p_a$ from its approximate near neighbors are in $S_1$, and therefore all the flipped coordinates of the worst-case query are in $S_1$, the probability of success for this query in hashing to size $d$ from the root is exactly $1$. In the remaining hashing from the size $d$ subset, the optimized algorithm has probability at least $d^{- \rho_o}$, where $\rho_o = \frac{\ln(1 - \frac{r}{d})}{\ln(1 - (1- \frac{1}{m}) \frac{r+1}{d})}$, from equation \eqref{eqn:precise_rho} in the proof of theorem \ref{thm:CorrectnessRuntime}. 

Once the dataset is of size $d$, according to the uniform (theoretical) LSH algorithm \citep{IndMot98}, the probability of success is exactly equal to $d^{-\rho_u}$, where $\rho_u = \frac{\ln(1 - \frac{r}{d})}{\ln(1 - \frac{r+1}{d})}$. The total probability of success for this worst-case query in uniform LSH is then, 

\begin{align}
    \Pr[success \mid \text{uniform}] &\leq d^{-\rho_u} (1 - \frac{r}{d})^{\ln(\frac{n}{d}) \ln^{-1}(1/\eps_1)}
\end{align}

The final advantage of our algorithm over uniform LSH follows from these formulae.

One fact that remains to show is that in the high-dimensional limit, the balances of the coordinates remain concentrated at $\eps$. 

Consider a node $v$ in the tree that was generated by hashing on $k_u$ coordinates. Consider an unhashed dimension $i \in [d]$. Let $f_{i,v}$ be the balance of coordinate $i$ at this node. As the dataset has independently drawn coordinates, the distribution of balances for $i$ is independent of the previous hashes, and so $f_{i,v} = \frac{1}{n_v} \cdot Binomial(n_v, \eps_i)$. Then, we can apply the standard Chernoff bound: 

\begin{align}
    Pr [ |f_{i,v} - \eps_i| > \delta] &\leq 2\exp\left(- \frac{1}{3} \eps_i n_v \delta^2 \right) \\ 
    &\leq 2\exp\left(- \frac{1}{3} \eps_2 d \delta^2 \right) \\
    &= \frac{1}{100d^{k+1}} 
\end{align}

Note there are a total of at most $d^k$ nodes and $d$ coordinates we must consider for $\leq k$ possible hashes. Thus, we set the failure probability to $\frac{1}{10d^{k+1}}$ so that the probability of success on all nodes is at least $(1 - \frac{1}{100d^{k+1}})^{d^{k+1}} \approx e^{-0.01} = 0.99$. Solving the previous equation for $d$ gives the requirement that $d \geq \frac{3(k+1) \ln (d) + 3\ln 200}{\delta^2 \eps_2}$. For fixed, $\delta$, $\eps_2$, and ${k \leq \frac{\ln(\frac{d}{n})}{\ln (1 - \eps_1)}}$, the left-hand-side grows faster with dimension than the right. Therefore, in the high-dimensional limit we can drive ${\delta \to 0}$ while maintaining a $0.99$ probability of success. 

\end{proof}

We also show improvement over the LSH forest algorithm. Recall that in this algorithm, for a given query, a coordinate is chosen uniformly at random, one at a time, until the current bucket has size less than or equal to 1. 

\begin{theorem}
Suppose we have a dataset drawn according to the aforementioned data model, with $d=100r$, $n = d^6$, $\eps_1 = 0.3,\eps_2 = 0.5$, $\alpha_1 = \alpha_2 = \frac{d}{2}$, but only one planted approximate near neighbor to $p_a$. Then with probability $0.99$ over the data distribution, trees constructed and queried with Algorithms \ref{alg:build_tree} and \ref{alg:query_tree} with query parameter $m=0$, exponent parameter $\rho \in (0.2,0.8)$ until the bucket is of size $d$, and then $\rho=0$ for the remainder of the tree until stopping condition $C=1$, where $k_u=\ln(\frac{n}{d}) \ln^{-1}(1/\eps_1)$, has $(1 - \frac{r}{d})^{-k_u}$ times greater success probability than uniform LSH trees for all queries (over the randomness of the algorithm and data model). 
\label{thm:improveUni2}
\end{theorem}

\begin{proof}
We first prove that the minimum-performing query to this dataset (for uniform LSH trees) is one with all coordinates flipped in $S_1$ (on the bits differentiating an approximate near neighbor from $p_a$). As there are $r+1$ coordinates for which $p_a$ differs from all other near neighbors, we must hash until the single coordinate that is not flipped in the worst-query, is flipped. The probability of this is $\frac{1}{d-k} \approx \frac{1}{d}$, where $k \ll d$ is the number of hashes chosen to get the dataset to size $d$, and increases to $\frac{1}{d-s}$ for $s$ additional hashes. Then, we will need at least $\frac{d}{2}$ additional hashes to get $O(1)$ probability of getting to a single point, using uniform hashing. (This is because $(1 - \frac{1}{d})\cdots (1 - \frac{2}{d}) < (1 - \frac{2}{d})^{\frac{d}{2}} = O(1)$). 

The probability of success for this query doing this is $(1 - \frac{r}{d})^\frac{d}{2} \approx e^{-r/2}$, which is clearly vanishing with $r$, and is greater than for all queries which are not designed to have $r$ of the $r+1$ differing coordinates flipped. Suppose we only flip $r - \ell + 1$ of these $r+1$ bits, for $\ell \geq 2$, then the probability of eventually hashing on one of the unflipped bits is $(1 - \frac{r}{d})^{\frac{d}{2\ell}} \approx e^{-r/2\ell}$. Suppose pessimistically the probability of success for other queries is $\frac{1}{d}$ (as good as randomly sampling points) times the probability of selecting one of the differing coordinates $e^{-r/2\ell}$. Suppose optimistically it is $e^{-r}$ for designed queries with $r$ of the $r+1$ bits flipped. Then we just require $\frac{1}{d} e^{-r/2\ell}  (1 - \frac{r}{3d})^{k_o} (1 - \frac{r}{d})^{-k_o} \geq e^{-r/2}$ for the designed query to be the true minimum, where $k_u$ is the number of hash functions needed to get to bucket size $d$ for optimized hashing. This inequality is true for large $r$ and $d=100r$, proving the worst-case query is as claimed.

Consider the first phase, where we hash the dataset until it is of size $d$. The probability of success for the optimized distribution on the worst-case query is exactly $1$, while for the uniform hash tree it is at most $(1 - \frac{r}{d})^{k_u}$, where $k_u = \ln(\frac{d}{n}) \ln^{-1}(\eps_1)$ (as we proved in the previous theorem). With very high probability, we will not have chosen the necessary differentiating bit to separate the approximate near neighbor from $p_a$. Therefore, the probability for the remainder of the tree for uniform is (with high probability) equal to that of our algorithm (as the datasets of size $d$ should be the same in expectation for both algorithms, given the independent coordinates assumption, the probability of success over the randomness in both the algorithm and the data model is equal for both algorithms).   
\end{proof}

\section{Uniform Distribution is Optimal for Independent Coordinates}
\label{sec:UniOptimal}

Suppose the data are drawn from the data model in section \ref{sec:D} without an additional planted cluster. Suppose further that instead of two groups, there are $M$ groups $S_i$ with $M \ll d$, balances $\eps_i \in (0,\frac{1}{2}]$, and cardinalities $d/M$. We also consider the limit where $r \ll d$. For simplicity of analysis, suppose we plant the point $0^d$ in the dataset. 

We consider a variant of the LSH tree where a fixed number of hash functions are selected from a chosen distribution $\pi$ until the dataset of size $n_0$, where $d \ll n_0$. In a standard LSH tree, where $w_i$ is the fraction of points remaining in the dataset after hashing for the $i$-th time, we select hashes ($k_o$ in total) such that:
\begin{align}
    \prod_{i=1}^{k_o} w_i &= \frac{n_0}{n} \\
    \iff \sum_{i=1}^{k_o} \ln w_i &= \ln \frac{n_0}{n}
\end{align}
As $k_o$ is a random variable, we instead consider the number of hash functions $k_o$ needed in expectation to reach the stopping size. In other words, we compute $k_o^*$ such that:

\begin{align}
    \mathbb{E} \left[ \sum_{i=1}^{k_o^*} \ln w_i \right] &= \ln \frac{n_0}{n}
\end{align}

In the LSH variant we propose here, we use this fixed $k_o^*$ number of hashes. 

\begin{theorem}
\label{thm:uniOptimal}
When the data are sampled according to the above data model, the uniform distribution is optimal for the worst-case query to the above LSH variant. 
\end{theorem}

\begin{proof}

We derive the exact value of the number of hash functions $k_o^*$. By assuming the coordinates are drawn independently, we use the linearity of expectation to derive:

\begin{align}
    \mathbb{E} \left[ \sum_{i=1}^{k_o^*} \ln w_i \right] &= \ln \frac{n_0}{n} \\
    \iff k_o^* &= \ln \frac{n_0}{n} \mathbb{E}^{-1} \left[ \ln w_i \right]  \\
    &= \ln \frac{n_0}{n} \left(\sum_{i=1}^M \pi_i \ln (1 - \eps_i) \right)^{-1}
\end{align}

The last step follows from two facts. First, we only need to consider distributions over coordinates that are marginally uniform across coordinates in a single group. This is because the worst queries to the dataset will be to the planted point $0^d$, whose balances are uniform across coordinates of a single group. Second, because we are in the high-dimensional limit (as in the previous section), when we hash on a single coordinate $i$, the fraction of points in the dataset that remains in the bucket is exactly $1 - \eps_i$, as this is the fraction of points that have a $1$ at coordinate $i$.

Consider a query at distance $r$ from $0^d$ with its $r$ coordinates flipped in an arbitrary group $j \in [M]$. To begin with, suppose $M=2$. The probability of success over the entire tree for this query is, using $k_o^*$ total hashes:

\begin{align}
    \Pr[success] &= \left(1 - \pi_j \frac{2r}{d} \right)^{k_o^*} \\ 
    \iff \ln \Pr[success] &= \ln (\frac{n_0}{n})  \ln \left(1 - \pi_2 \frac{2r}{d} \right) \mathbb{E}^{-1} \left[ \ln w_i \right] \\
    &\approx \ln (\frac{n}{n_0}) \frac{2\pi_2 r}{d} \mathbb{E}^{-1} \left[ \ln w_i \right]   \\
    &\propto \frac{\pi_2}{\pi_1 \ln w_1 + \pi_2 \ln w_2}  \\
    &= \frac{\pi_2}{(1-\pi_2)\ln w_1 + \pi_2 \ln w_2}
\end{align}

As the logarithm is increasing, we can compute the derivative of the RHS to understand the optimal setting of $\pi_2 = 1 - \pi_1$. Doing so, we find that the derivative is $\frac{d}{d \pi_2} (\text{RHS}) = \frac{\ln w_1}{((1-\pi_2)\ln w_1 + \pi_2 \ln w_2)^2} < 0$. Therefore, the probability of success increases by decreasing $\pi_2$. Further, if the query has its bits flipped on group $S_1$ instead, the probability of success is also decreasing in $\pi_1$. Therefore, the optimal distribution decreases $\pi_2$ until the probability of success for both types of queries are equal. Setting these two query probabilities to be equal:
\begin{align}
    \frac{\pi_2}{(1-\pi_2)\ln w_1 + \pi_2 \ln w_2} = \frac{\pi_1}{(1-\pi_2)\ln w_1 + \pi_2 \ln w_2}
\end{align}
we derive that $\pi_1 = \pi_2$, i.e. the uniform distribution is optimal.

Generalizing to many groups (for increasing, non-positive functions $F_j(\pi_j)$):
\begin{align}
    \ln \Pr[success] &\approx \ln (\frac{n}{n_0}) \frac{\pi_j Mr}{d} \mathbb{E}^{-1} \left[ \ln w_i \right]   \\
    &\propto \frac{\pi_j}{\sum_{i=1}^M \pi_i \ln w_i}  \\
    &= \frac{\pi_j}{F_j(\pi_j) + \pi_j \ln w_j} \label{eqn:subF}
\end{align}

As the logarithm is increasing, we can compute the derivative of the RHS to understand the optimal setting of $\pi_j$. Doing so, we find that the derivative is $\frac{d}{d \pi_j} (\text{RHS}) = \frac{F_j(\pi_j) - \pi_j F_j^\prime(\pi_j) }{(\pi_j \ln w_j - F_j(\pi_j-1) )^2} < 0$. Again, the probability of success is decreasing in $\pi_j$. Further, the denominator is the same regardless of where the query's bits are flipped (i.e. which group is chosen to flip). So, for a fixed distribution, the log of the probability of success is proportional to the probability of choosing the group with bits flipped for that query. In this independent case, there are essentially $M$ possible types of queries - one with bits flipped entirely in each one of the groups. Suppose that for a chosen distribution, the probability of success is higher for queries in one group versus another. Then, by the derivative argument above, we can increase the success probability for the worse query by moving weight from that group to the other. Therefore, any distribution that has this inequity is not optimal. Therefore, the optimal distribution is such that for all $j,k \in [M]$:
\begin{align}
    \frac{\pi_k}{\sum_{i=1}^M \pi_i \ln w_i} = \frac{\pi_j}{\sum_{i=1}^M \pi_i \ln w_i}
\end{align}

Hence, the uniform distribution is again optimal.
\end{proof}

\section{Additional Experiments}
\label{sec:AdditionalExperiments}
We performed a variety of additional experiments to demonstrate our algorithm's effectiveness. (1) We performed a set of experiments on the entire MNIST dataset and a 100,000-point subset of the ImageNet dataset, (2) we measured the query times of our algorithm on all subsets. All experiments show our algorithm can perform much better than uniform LSH forests.

For both large datasets, we set the stopping bucket size to 10 and c=1, while MNIST used radius $r=3$ and ImageNet used $r=2$. For ImageNet, we performed experiments on the first 100k images of the 8x8 training subset of the dataset. The images were binarized with pixel threshold value of 70, while MNIST was binarized with threshold 1. The MWU parameter beta was set to 0.4 in all experiments on these datasets, and an aggressive update method was used where the optimization returned the most recent hash strategy rather than the average. We also note that for the small subsets, the query strategies were chosen against the average response of the hash players and were played simultaneously, although this does not affect the solution of the game, and may only slow down convergence. 

For experiments in both large datasets, measurements were collected from 8 trees formed by our optimized algorithm and compared to uniform LSH trees. We sampled hash functions uniformly until the buckets were of size 700 for MNIST and of size 1000 for ImageNet, then we ran the game with exponent $\rho=1$ for 500 and 2000 rounds on MNIST and ImageNet, respectively. Two queries were generated at random for each point of the datasets, meaning 120,000 queries were measured for MNIST and 200,000 queries were measured for ImageNet in total. It is straightforward to obtain additional improvement in recall/query-times by performing the optimizations for more rounds and by optimizing at all nodes in the tree (rather than at just those with fewer than 700 points). 

On querying, we measured the time until the near neighbor was returned for a given query/NN pair using the “time” library for Python. We measured the average and bottom tenth percentile of success probabilities, which is the average recall over the bottom tenth of success probabilities for random queries. This is a proxy for the minimum success probability, as for some trees the success probability was too small to be measured. This occurs because we are not using pivots in our experiments. Our algorithms has far shorter query times for both datasets and all subsets, particularly for the queries with the largest query times (Table \ref{table:ImageNet_times}, \ref{table:large_ImageNet_times}). Improvements in query times are consistent with improvements in the recall of our algorithm over uniform hashing (Table \ref{table:large_ImageNet_succ_probs}).

In Figure \ref{fig:distributions}, we extracted the hash distributions over coordinates produced by our optimization at the root of the LSH trees on the MNIST subset. After scaling the distribution by the inverse of the mean and centering the scaled distribution to have mean 0, we constructed heatmaps for two sets of optimization parameters. The heatmaps show that the optimized distributions place more weight on the coordinates in the center of the image (where there is more variation among images). 

\comment{
The Recall/QPS curves were constructed by averaging query times/recall over these queries. Recall was measured as the fraction of near neighbors that were recovered for a given query with a fixed number of trees (assuming each point has a single near neighbor in the dataset). Queries-per-second was calculated by taking the inverse of query time. Because the recall was too high to resolve differences among tree types for our settings, we increased the query radius $r$ to 6 at query time for both the ImageNet and MNIST datasets. The curves being shifted up and to the right for our algorithm versus uniform hashing indicates the trees have both higher recall and throughput (Figure \ref{fig:recall_QPS_tradeoff}).
}

\onecolumn

\begin{center}
\begin{table}[ht]
  \caption{Query Times on Random Queries to Small Subsets.}
  \centering
  \vskip 0.15in
  \begin{tabular}{llll}
    \toprule
    Dataset & Parameters & Maximum ($s$) & Average ($s \times 10^{-5}$)\\
    \midrule
    IN & Uniform & 0.0013 & 1.32 \\
    IN & $\rho=1$, $T=3000$, $\beta=0.68$ & \textbf{0.00046} & \textbf{0.737} \\
    \midrule
    MNIST & Uniform & 0.0012 & 4.03 \\
    MNIST & $\rho=1$, $T=3000$, $\beta=0.68$ & \textbf{0.00048} & \textbf{1.10} \\
    MNIST & $\rho=0.83$, $T=3000$, $\beta=0.68$ & 0.0011 & 1.29 \\
    MNIST & $\rho=0.25$, $T=1600$, $\beta=0.88$ & 0.0033 & 2.15 \\
    MNIST & $\rho=0.1$,  $T=1600$, $\beta=0.88$ & 0.0027 & 3.06 \\
    \bottomrule
  \end{tabular}
\label{table:ImageNet_times}
\vskip -0.1in
\end{table}

\begin{table}[ht!]
  \caption{Query Times on Random Queries to Large Datasets.}
  \vskip 0.15in
  \centering
  \begin{tabular}{llll}
    \toprule
    Dataset & Hash Distribution & 90th Percentile ($s \times 10^{-4}$) & Average ($s \times 10^{-4}$)\\
    \midrule
    IN & Uniform & 4.91 & 2.43 \\
    IN & Alg. & \textbf{3.65} & \textbf{2.02} \\
    \midrule
    MNIST & Uniform & 13.27 & 7.84 \\
    MNIST & Alg. & \textbf{8.66} & \textbf{5.24} \\
    \bottomrule
  \end{tabular}
\label{table:large_ImageNet_times}
\vskip -0.1in
\end{table}

\begin{table}[ht!]
    
  \caption{Success Probability on Random Queries to Large Datasets.}
  \label{sample-table}
  \vskip 0.15in
  \centering
  \begin{tabular}{llll}
    \toprule
    Dataset & Hash Distribution & Average Bottom 10\% & Average \\
    \midrule
    IN & Uniform & 0.117 & 0.68 \\
    IN & Alg. & \textbf{0.127} & \textbf{0.73} \\
    \midrule
    MNIST & Uniform & 0.51 & 0.830\\
    MNIST & Alg. & \textbf{0.66} & \textbf{0.893} \\
    \bottomrule
  \end{tabular}
\label{table:large_ImageNet_succ_probs}
\vskip -0.1in
\end{table}

\comment{
\begin{figure}[ht]
\centering
\begin{subfigure}{.4\textwidth}
  \centering
  \includegraphics[scale=0.7]{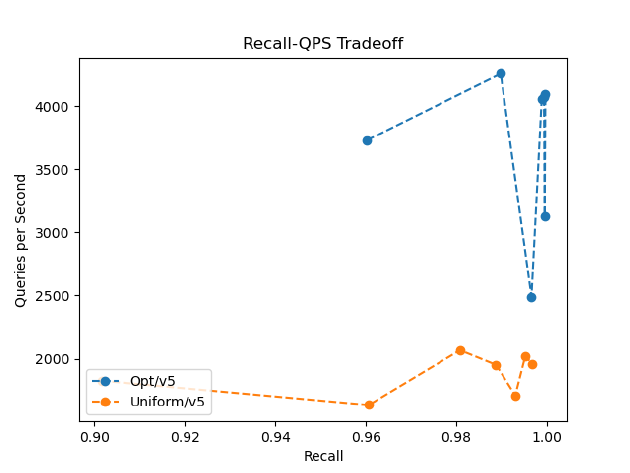}  
  \caption{MNIST}
  \label{fig:sub-second}
\end{subfigure}
\begin{subfigure}{0.4\textwidth}
  \centering
  \includegraphics[scale=0.7]{images/whole_imnet_qps_recall.pdf}  
  \caption{ImageNet}
  \label{fig:sub-second}
\end{subfigure}

\centering
\captionsetup{justification=centering}
\caption{Recall/Query-per-second Tradeoff Curves (orange - Uniform, blue - our Algorithm)}

\label{fig:recall_QPS_tradeoff}

\end{figure}

\begin{figure}[ht]
\begin{center}
\includegraphics[scale=1.7]{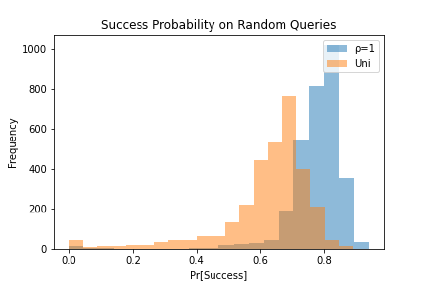}
\caption{Histogram of Recall for Queries to the smaller ImageNet Subset}
\label{fig:imagenet_histo}
\end{center}
\end{figure}
}
\end{center}

\section{Implementation Details}
\label{sec:implementationdeets}
\label{sec:F}

A link to the experiment code repository can be found at (https://anonymous.4open.science/r/instance-optimal-lsh-51DF/README.md). The experiments were implemented in C++, and compiled with g++-5 using the -march=native and -O3 flags for improved runtime. In addition, our implementation was highly parallelized using OpenMP pre-proccessor directives. Efficient matrix/vector computation was done with the Eigen library for C++ \citep{eigenweb}. The experiments were performed on an Intel(R) Xeon(R) W-2155 CPU @ 3.30GHz with 65 GB of RAM (all 20 physical cores were used for the experiment). Query times for the small subsets were measured on a 2.3 GHz Dual-Core Intel Core i5 with 8GB of RAM. The runtime to generate 110 trees with 3000 game rounds varied, but took on average 40 hours to complete with these hardware specs.

\end{document}